\documentclass[letterpaper,11pt]{article}
\usepackage{amsfonts}
\usepackage{amsmath,amsthm}
\usepackage{amssymb,amscd}
\usepackage{enumerate}
\usepackage{cite}

\makeatletter
\newcommand{\rmnum}[1]{\romannumeral #1}
\newcommand{\Rmnum}[1]{\expandafter\@slowromancap\romannumeral #1@}
\makeatother

{}
\newtheorem{theorem}{Theorem}[section]

\newtheorem{corollary}[theorem]{Corollary}
\newtheorem{lemma}[theorem]{Lemma}
\newtheorem{proposition}[theorem]{Proposition}

\theoremstyle{definition}
\newtheorem{definition}{Definition}[section]

\theoremstyle{remark}
\newtheorem{remark}{Remark}[section]

\swapnumbers \makeatletter
\def\ps@mystyles{    \def\@evenhead{\hfil{\sl\leftmark}\hfil}    \def\@oddhead{\hfil{\sl\rightmark}\hfil}    }
\makeatother

\begin{document}
\title{The transition probability  and the probability for the left-most particle's position of the $q$-TAZRP}

\author{\textbf{Marko Korhonen\footnote{marko.korhonen@helsinki.fi; Department of mathematics and statistics, University of Helsinki, Finland}}~~~and~~~\textbf{Eunghyun Lee\footnote{eunghyun@gmail.com; Centre de recherches math\'{e}matiques (CRM), Universit\'{e} de Montr\'{e}al, Canada}}}
\date{}
\maketitle

\begin{abstract}
\noindent We treat the $N$-particle ZRP whose jumping rates satisfy a certain condition. This condition is required to use the Bethe ansatz and the resulting model is the $q$-boson model that appeared  in [J. Phys. A, \textbf{31} 6057--6071 (1998)] by Sasamoto and Wadati or  the $q$-TAZRP in \textit{Macdonald processes} by Borodin and Corwin. We find the explicit formula of the transition probability of the $q$-TAZRP via the Bethe ansatz. By using the transition probability we find the probability distribution of the left-most particle's position at time $t$. To find the probability for the left-most particle's position we find a new identity corresponding to Tracy and Widom's identity for the ASEP in [Commun. Math. Phys., \textbf{279} 815--844 (2008)]. For the initial state that all particles occupy a single site, the probability distribution of the left-most particle's position at time $t$ is represented by the contour integral of a determinant.
\end{abstract}
\section{Introduction}
 The zero range process (ZRP)  describes the system of indistinguishable particles on a countable set governed by the following law : if a site $x$ is occupied by $k$ particles, one of $k$ particles leaves $x$ at the rate $g(k)$. When the particle leaves $x$, it chooses a target site $y$ with probability $p(x,y)$ and jumps to $y$. In other words, if $x$ is occupied by $k$ particles, then one of $k$ particles at $x$ jumps to $y$ at the rate $g(k)p(x,y)$. Since the ZRP was introduced by Spitzer \cite{Spitzer}, this model has been very popular   in both  mathematics and physics communities for a long time along with the asymmetric simple exclusion process (ASEP).  It is well known that under the assumption that $g(k) >0$ for $k >0$,  $g(0) = 0$ and $\sup_k |g(k+1) - g(k)| < \infty$,  the ZRP is a well-defined Markov process \cite{Andjel}.  In this paper, we consider the $N$-particle  totally asymmetric zero range process (TAZRP) on $\mathbb{Z}$ with a special form of $g(k)$ and we assume that $p(x,y) = 1$ if $y=x+1$.
\subsection{Bethe ansatz applicability}
Since the  work of Sch\"{u}tz \cite{Schutz}, the coordinate Bethe ansatz has been successfully used to find the transition probabilities of some stochastic particle models with countable state spaces \cite{Borodin2,Chatterjee,Lee2,Lee3,Rakos,TW1,TW6}.  The time evolution of a stochastic particle model is governed by the master equation, which is a system of first-order differential equations, and  for each state $X$ the corresponding differential equation describes the time-evolution of the probability $P(X;t)$ that the system is at  a state $X$ at time $t$.  In a matrix form, the master equation is given by
\begin{equation}
\frac{d}{dt}P(X;t) = H P(X;t) \label{masterEQ}
\end{equation}
where $H$ is the matrix  generating the time evolution of the probability distribution $P(X;t)$. The Bethe ansatz provides a way of obtaining an eigenfunction of  $H$. In a very recent paper  by Borodin, Corwin, Petrov and Sasamoto \cite{Borodin4}, the spectral theory for $H$ of the model we are considering in this paper has been developed. In using the Bethe ansatz, it is an essential point that the system of differential equations (\ref{masterEQ}) can be replaced by a single type of differential equation along with  a single type of \textit{boundary} conditions, and the solution of the set of  these differential equations and boundary conditions is called the \textit{Bethe ansatz solution}.   This requirement for the Bethe ansatz limits the \textit{disorder} of the transition rates to a certain condition \cite{Ali,Ali2,Povol2,Povol,Sasamoto2,SasaWada1998Oct}, and so the Bethe ansatz method is applicable to only a certain type of the TAZRP.   It is interesting that transition rates of the Bethe ansatz applicable models are constants or expressed by $q$-numbers in many cases and this type of rates also appears in the structure of Macdonald processes \cite{Borodin3,Corwin}. For the TAZRP, $g(k)$ should be $[k]_q$ where $q= g(2)-1$ for the use of the Bethe ansatz and the TAZRP with jumping rates $g(k) = [k]_q$ was referred to as the \textit{$q$-TAZRP} \cite{Borodin3,Borodin5}. In fact, the $q$-TAZRP is originated from the $q$-boson totally asymmetric diffusion model by Sasamoto and Wadati \cite{Sasamoto2}, and Povolotsky \cite{Povol} obtained the explicit form of $g(k)=[k]_q$ in the direct language of particle systems. In addition, the $q$-TASEP with $N$ particles is interpreted as  the $q$-TAZRP on $N$ sites and the $q$-TASEP was intensively studied as a model of Macdonald processes \cite{Borodin3,Borodin5}.
\subsection{Transition probability and the left-most particle's position at time $t$}
The transition probability is a solution of the master equation and it must also satisfy the initial condition. In stochastic particle models, the transition probability often plays an important role of a starting point to study probability distributions of meaningful random variables and further their asymptotic behaviors \cite{Borodin6,Lee,TW2,TW3}.   Given the initial state $Y$, we denote the transition probability by $P_Y(X;t)$. Some known transition probabilities of $N$-particle stochastic models on $\mathbb{Z}$ obtained by the Bethe ansatz are given by the sum of $N!$ $N$-dimensional contour integrals \cite{Borodin2,Chatterjee,Lee2,Lee3,Rakos,Rakos2,Schutz,TW1}. These transition probabilities are obtained by the procedure of having the Bethe ansatz solution satisfy the initial condition.
The transition probability of the $q$-TAZRP we find in this paper is similar to the other known transition probabilities in the sense that it is written as $N!$ $N$-dimensional contour integrals and the integrand is originated from the Bethe ansatz solution. However, in the transition probability of the $q$-TAZRP  the contour integral is multiplied by a state-dependent weight $W_X$. (See Theorem \ref{theorem1}.) These weights $W_X$ are known to give the stationary measure of the general ZRP up to a constant \cite{Evans}.  Compared to other models \cite{Lee3,TW1}, a novel point of proving the transition probability of the $q$-TAZRP is that it needs an identity related to a $[k]_q$. (See Lemma \ref{identity1}.) The transition probability of the $q$-TAZRP was conjectured in \cite{Povol3} and found by using a different approach \cite{Borodin4}.  The transition probability in Theorem \ref{theorem1} can be a starting point to compute various probability distributions. Let us denote   a state $X$ of $N$-particle system  by $X=(x_N,\cdots,x_1) $ where $x_N \leq \cdots \leq x_2 \leq x_1$ representing positions of particles, and the position of the $m$th right-most particle at time $t$ by $x_m(t)$. Then  $\mathbb{P}_Y(x_N(t) =x)$  where $\mathbb{P}_Y$ is the probability measure of the process with the initial state $Y$ is obtained by
\begin{equation}
\mathbb{P}_Y(x_N(t) =x) =\sum_{k_1,\cdots,k_{N-1}=0}^{\infty} P_Y(x,x+k_1,\cdots,x+k_1+\cdots + k_{N-1};t) \label{IntroEQ1}
\end{equation}
where the series can be shown to be convergent geometric series. In the ASEP, the indices $k_1,\cdots,k_{N-1}$ start from 1 because of the exclusion property of the model  and the right-hand side of (\ref{IntroEQ1}) is just the sum of $N!$ terms after summing over all $k_i$s because the transition probability is the sum of $N!$ contour integrals. In the work for the ASEP \cite{TW1}, Tracy and Widom found an identity (See (1.6) in \cite{TW1})  that states that the antisymmetrization of a certain  function of $N$ variables is written as the Vandermonde determinant times a symmetric function which has a special form. The point is that the symmetric function is factorized by variables and $\mathbb{P}_Y(x_N(t) =x)$ is written as a contour integral whose integrand is a determinant. Tracy and Widom's identity also appears in the two-sided PushASEP \cite{Lee3}. However, in the $q$-TAZRP, the right-hand side of (\ref{IntroEQ1}) is written as the sum of $N!$ times $2^{N-1}$ terms. This is due to the fact that the weight $W_X$ for $X=(x,x+k_1,\cdots,x+k_1+\cdots + k_{N-1})$  varies, depending on whether each $k_i$ is zero or not. What is surprising is that this sum of $N!$ times $2^{N-1}$ terms is also written as the Vandermonde determinant times a symmetric function that is factorized by variables as in Tracy and Widom's identity (See Proposition \ref{NewIdentity}.) and the right-hand side of (\ref{IntroEQ1}) is also written as a contour integral of a determinant (See Theorem \ref{THEO} and Corollary \ref{COROLL}.).
\subsection{Outline}
In section 2 we briefly review the Bethe ansatz applicability of the ZRP originally discussed in \cite{Povol} and find the transition probability of the $q$-TAZRP. Also, we discuss about a mapping between the $N$-particle $q$-TAZRP and a variant of $N$-particle $q$-TASEP. In section 3 we prove the new identity that emerges in the computation of $\mathbb{P}_Y(x_N(t) =x)$ and then provide $\mathbb{P}_Y(x_N(t) =x)$ and its formulas when the initial state is that a single site is occupied by all particles.
\section{Transition probability of the $q$-TAZRP}
\noindent  We simply assume that $g(1)=1$ in this paper and let us denote the state space by $ \textbf{S}.$   A state $X= (x_N,\cdots,x_1)$ which is represented by positions of particles also can be represented by coordinates of occupied sites and the number of particles at each occupied site.  Let $\Lambda(X)$ be the set of all occupied sites of state $X$ and $\mathbf{x}_i$ be the $i$th right-most occupied site. Let  $\eta(\mathbf{x}_i)$ be the number of particles occupying  $\mathbf{x}_i$, and if $\mathbf{x}_i$ is occupied by $n_i$ particles, we denote it by $(\mathbf{x}_i)^{n_i}$. Hence, an alternate representation of  state $X$ is
$$X=\big((\mathbf{x}_m)^{n_m},\cdots,(\mathbf{x}_1)^{n_1} \big)$$
 with $\sum n_i = N$ and $m \leq N.$
\subsection{Preliminary}
  We briefly review the procedure of obtaining the $q$-TAZRP model from the general ZRP. The discussion in this subsection is based on \cite{Povol}.
\subsubsection{2-particle ZRP} The master equation (\ref{masterEQ}) consists of the following three types of equations.
\begin{eqnarray}
 \frac{d}{dt}P(x_2,x_1;t) &=& P(x_2-1,x_1;t) + P(x_2,x_1-1;t) -2 P(x_2,x_1;t) \label{eq1} \\
 & &\hspace{5cm} \textrm{for}~~x_2 <x_1-1  \nonumber \\
 \frac{d}{dt}P(x,x+1;t) &=& P(x-1,x+1;t) + g(2)P(x,x;t) -2 P(x,x+1;t) \label{eq2} \\
 \frac{d}{dt}P(x,x;t) &=& P(x-1,x;t) - g(2)P(x,x;t) \label{eq3}
\end{eqnarray}
Assume that $u^0(x_2,x_1;t)$ satisfies
 \begin{equation}
\frac{d}{dt}~u^0(x_2,x_1;t) = u^0(x_2-1,x_1;t) + u^0(x_2,x_1-1;t) -2 u^0(x_2,x_1;t) \label{eq4}
\end{equation}
for any $(x_2,x_1) \in \mathbb{Z}^2$ so that    $u^0(x_2,x_1;t)$ obviously satisfies  (\ref{eq1}). When $x_2=x$ and $x_1 =x+1$, (\ref{eq4}) is
\begin{equation}
\frac{d}{dt}~u^0(x,x+1;t) = u^0(x-1,x+1;t) + u^0(x,x;t) -2 u^0(x,x+1;t), \label{eq5}
\end{equation}
and if we define
\begin{equation}
\label{def} u(x_1,x_2;t) :=
\begin{cases}
 u^0(x_1,x_2;t) & \textrm{if}~~x_1 \neq x_2\\
\frac{1}{g(2)}u^0(x_1,x_2;t) & \textrm{if}~~x_1 = x_2,
\end{cases}
\end{equation}
then we see that $u(x_1,x_2;t)$ satisfies (\ref{eq1}) and (\ref{eq2}).
When $x_2=x_1=x$, (\ref{eq4}) is
\begin{equation}
\frac{d}{dt}~u^0(x,x;t) = u^0(x-1,x;t) + u^0(x,x-1;t) -2 u^0(x,x;t), \label{eq6}
\end{equation}
and if $u^0(x_2,x_1;t)$ satisfies
\begin{equation}
u^0(x,x-1;t) = (g(2)-1)u^0(x-1,x;t) - (g(2) -2)u^0(x,x;t) \label{boundary},
\end{equation}
then we see $u(x_1,x_2;t)$ satisfies (\ref{eq3}) as well. Hence, the solution of the original system of three types of differential equations is obtained by the definition (\ref{def}) and the solution of (\ref{eq4}) that satisfies (\ref{boundary}).
\subsubsection{3-particle ZRP}
Assume that  $u^0(x_1,x_2,x_3;t)$  satisfies
 \begin{eqnarray}
\frac{d}{dt}~u^0(x_3,x_2,x_1;t) &=& u^0(x_3-1,x_2,x_1;t) + u^0(x_3,x_2-1,x_1;t) + u^0(x_3,x_2,x_1-1;t) \nonumber\\
                                & & \hspace{3cm} -3 u^0(x_3,x_2,x_1;t) \label{N3eq}
\end{eqnarray}
for any $(x_3,x_2,x_1) \in \mathbb{Z}^3$. Extending  (\ref{def}) as
\begin{equation}
\label{def3} u(x_3,x_2,x_1;t) :=
\begin{cases}
 u^0(x_3,x_2,x_1;t) & \textrm{if}~~x_{i+1} < x_{i}~\textrm{for}~i=1,2~\textrm{or}~(x_3,x_2,x_1) \notin \mathbf{S},\\
\frac{1}{g(2)}u^0(x_3,x_2,x_1;t) & \textrm{if}~~x_3 = x_2 < x_1 ~\textrm{or}~x_3 < x_2=x_1, \\
\frac{1}{g(2)g(3)} u^0(x_3,x_2,x_1;t) & \textrm{if}~~x_1=x_2=x_3
\end{cases}
\end{equation}
and  the \textit{boundary} condition
(\ref{boundary}) as
\begin{eqnarray}
u^0(x,x-1,x_1;t) &=& (g(2)-1)u^0(x-1,x,x_1;t) - (g(2) -2)u^0(x,x,x_1;t), \label{bd3}  \\
u^0(x_3,x,x-1;t) &=& (g(2)-1)u^0(x_3,x-1,x;t) - (g(2) -2)u^0(x_3,x,x;t), \label{bd4}
\end{eqnarray}
then we see that $u(x_3,x_2,x_1;t)$ satisfies all types of differential equations in (\ref{masterEQ})  except
\begin{equation}
 \frac{d}{dt}~P(x,x,x;t) = P(x-1,x,x;t) - {g(3)}P(x,x,x;t). \label{M333eq}
\end{equation}
When $x_1=x_2=x_3=x$, (\ref{N3eq}) is
\begin{eqnarray}
\frac{d}{dt}~u^0(x,x,x;t) &=& u^0(x-1,x,x;t) + u^0(x,x-1,x;t) + u^0(x,x,x-1;t) \nonumber\\
                                & & \hspace{3cm} -3 u^0(x,x,x;t), \label{N33eq}
\end{eqnarray}
which is written as
\begin{equation}
 \frac{d}{dt}~u(x,x,x;t) = \frac{g(2)^2-g(2)+1}{g(3)}u(x-1,x,x;t) - {g(3)}u(x,x,x;t) \label{N333eq}
\end{equation}
by (\ref{def3}), (\ref{bd3}) and (\ref{bd4}). Hence, if we set $g(2)^2-g(2)+1 = g(3)$, (\ref{N333eq}) is in the exactly same form as (\ref{M333eq}). Hence, provided that $g(2)^2-g(2)+1 = g(3)$, all types of differential equations in (\ref{masterEQ}) are replaced by (\ref{N3eq}) with (\ref{bd3}) and (\ref{bd4}).
\subsubsection{ $N$-particle ZRP}
We extend the argument for $N=2,3$ to the $N$-particle ZRP. Assume  that
$u^0(x_N,\cdots,x_1;t)$ satisfies
\begin{eqnarray}
\frac{d}{dt}~u^0(x_N,\cdots,x_1;t) &=& \sum_{i=1}^N u^0(x_N,\cdots,x_{i+1}, x_i-1,x_{i-1},\cdots,x_1;t) \nonumber\\
                                   & & \hspace{3cm}-N u^0(x_N,\cdots,x_1;t) \label{Neq}
\end{eqnarray}
for any $(x_N,\cdots,x_1) \in \mathbb{Z}^N$.  It is reasonable to extend (\ref{def3})  as follows.
\begin{definition}\label{DEF}
\begin{equation}
\label{def2} u(x_N,\cdots,x_1;t) :=
\begin{cases}
u^0(x_N,\cdots,x_1;t) & \textrm{on}~~\mathbb{Z}^N \setminus \mathbf{S},\\
\frac{1}{\prod_{\mathbf{x} \in \Lambda(X)}\prod_{k=1}^{\eta(\mathbf{x})}g(k)}u^0(x_N,\cdots,x_1;t) & \textrm{on}~~ \mathbf{S}.
\end{cases}
\end{equation}
\end{definition}
\noindent Also we extend (\ref{bd3}) and (\ref{bd4}) to
\begin{eqnarray}
& &\hspace{1cm}u^0(x_N,\cdots,x_{i+1},x_i,x_i-1,x_{i-2},\cdots,x_1;t) \nonumber \\
& =& (g(2)-1)u^0(x_N,\cdots,x_{i+1},x_i-1,x_i,x_{i-2},\cdots,x_1;t) \nonumber\\
& & -~ (g(2) -2)u^0(x_N,\cdots,x_{i+1},x_i,x_i,x_{i-2},\cdots,x_1;t)~~\textrm{for}~~2 \leq i \leq N. \label{boundaryN}
\end{eqnarray}
We shall show that $u(x_N,\cdots,x_1;t)$ defined by (\ref{def2}) satisfies all types of differential equations in the system of equations (\ref{masterEQ}) for $N$ if $g(k)= [k]_q$ where $q=g(2)-1$.\\ \\
 \noindent The state space $\textbf{S}$ is partitioned into  $\textbf{S}_1,\textbf{S}_2$ and $\textbf{S}_3$ where
\begin{eqnarray*}
\textbf{S}_1 &=& \{(x_N,\cdots,x_1)  : x_i < x_{i-1} - 1 ~\textrm{for any}~i\}, \\
\textbf{S}_2 &=& \{(x_N,\cdots,x_1)  : x_i < x_{i-1} ~\textrm{for any}~i~\textrm{and}~x_l = x_{l-1} - 1 ~\textrm{for some}~l \},~\textrm{and}  \\
\textbf{S}_3 &=& \{(x_N,\cdots,x_1)  : x_l = x_{l-1} ~\textrm{for some}~l\}.
\end{eqnarray*}
Since $u(x_N,\cdots,x_1;t) = u^0(x_N,\cdots,x_1;t)$ on $\mathbf{S}_1$, it is obvious that $u(x_N,\cdots,x_1;t)$ satisfies (\ref{Neq}) which is in the same form as the differential equation for $P(x_N,\cdots,x_1;t)$ on $\mathbf{S}_1$. For $\mathbf{S}_2$, assuming that $x_l = x_{l-1}-1$ for only one $l$ without loss of generality, we can show that $u(x_N,\cdots,x_1;t)$ satisfy  differential equations for $P(x_N,\cdots,x_1;t)$  on $\mathbf{S}_2$ in the same way as in $N=2$ or $N=3$ case. Now, we show that $u(x_N,\cdots,x_1;t)$ satisfies differential equations for $P(x_N,\cdots,x_1;t)$ on $\mathbf{S}_3$.
\begin{lemma}\label{lemma1}
 For a fixed $K$ and $n$ with $K+n \leq N$, consider $X=(x_N,\cdots,x_1)$ with $x_{K+1}=\cdots =x_{K+n} =x$. Denote by $u^0(X_i;t)$ the one obtained by replacing $(K+n-i+1)$th right-most argument in $X$ by $(x-1)$ where $1 \leq i \leq n$. Then, for each $i$
\begin{equation*}
u^0(X_i;t) = q^{i-1}u^0(X_1;t) - (q^{i-1}-1)u^0(X;t).
\end{equation*}
\end{lemma}
\begin{proof}
The statement is inductively proved by using (\ref{boundaryN}).
\end{proof}
\noindent In order to show that $u(x_N,\cdots,x_1;t)$ satisfies differential equations for $P(x_N,\cdots,x_1;t)$ on $\mathbf{S}_3$, it suffices to show the following proposition.
\begin{proposition}
Consider  $X =(x_N,\cdots,x_1)$ with $x_{K+1},\cdots,x_{K+n} =x$ for a fixed $K$ and $n$ with $K+n \leq N$ and suppose that all other sites except $x$ in $\Lambda(X)$ are occupied by only one particle. If $g(n) = [n]_q$, then
 $u(x_N,\cdots,x_1;t)$ defined by (\ref{def2}) satisfies  the differential equation for $P(X;t)$.
\end{proposition}
\begin{proof}
By using the notation in Lemma \ref{lemma1}, the equation (\ref{Neq}) is written as
\begin{eqnarray}
\frac{d}{dt}~u^0(X;t) &=& \sum_{ j \neq K+1,\cdots,K+n } u^0(x_N,\cdots,x_{j+1}, x_j-1,x_{j-1},\cdots,x_1;t) \nonumber \\
   & & \hspace{2cm}+~ \sum_{i=1}^n u^0(X_i;t) -N u^0(X;t). \label{propo1}
\end{eqnarray}
Then  by Lemma \ref{lemma1} and (\ref{def2}), (\ref{propo1}) is equivalent to
\begin{eqnarray*}
\frac{d}{dt}~u(X;t) &=&  \sum_{ j \neq K+1,\cdots,K+n}\alpha_j u(x_N,\cdots,x_{j+1}, x_j-1,x_{j-1},\cdots,x_1;t)  \\
& &\hspace{1cm}+~ \frac{1}{g(n)}\frac{1-q^n}{1-q} \alpha_{K+n} u(X_1;t) -\Big(\frac{1-q^n}{1-q} +N-n\Big) u(X;t) \nonumber
\end{eqnarray*}
where $\alpha_j$ is the jumping rate for state $(x_N,\cdots,x_{j+1}, x_j-1,x_{j-1},\cdots,x_1) $ and this is exactly in the same form as the differential equation for $P(X;t)$ if $g(n) = [n]_q$.
\end{proof}

\subsection{Transition probability}
\noindent We find the solution $u^0(x_N,\cdots,x_1;t)$ of (\ref{Neq}) that satisfies (\ref{boundaryN}). Then $u(x_N,\cdots,x_1;t)$ provides the solution of (\ref{masterEQ}) through Definition \ref{DEF}. Moreover, if $u(x_N,\cdots,x_1;t)$ satisfies the initial condition on $\mathbf{S}$, $u(x_N,\cdots,x_1;t)$ restricted on $\mathbf{S}$ is the transition probability, $P_Y(X;t)$. The solution $u^0(x_N,\cdots,x_1;t)$ of (\ref{Neq}) that satisfies (\ref{boundaryN}) is obtained by the standard method, the Bethe ansatz.  Define the \textit{$S$-matrix}
\begin{equation}
S_{\beta\alpha} := S_{\beta\alpha} (z_{\alpha},z_{\beta})  = -\frac{z_{\beta} - qz_{\alpha} - (1-q)z_{\alpha}z_{\beta}}{z_{\alpha} - qz_{\beta} - (1-q)z_{\alpha}z_{\beta}}, ~~z_{\alpha},z_{\beta} \in \mathbb{C}. \label{SMatrix}
\end{equation}
and
\begin{equation}
 \epsilon (z_i) :=\frac{1}{z_i} - 1. \label{energe}
\end{equation}
The Bethe ansatz solution for  (\ref{Neq}) with (\ref{boundaryN}) is given by
\begin{equation}
u^0(x_N,\cdots,x_1;t) = \sum_{\sigma \in \mathbb{S}_N} {A'}_{\sigma} \prod_{i=1}^N z_{\sigma(i)}^{x_i}e^{\epsilon(z_i) t} \label{solution}
\end{equation}
for any nonzero $z_1,\cdots,z_N \in \mathbb{C}$ where ${A'}_{\sigma}$ satisfies
\begin{equation}
{A'}_{\sigma}= {A'}_{\textrm{Id}} \prod_{(\alpha,\beta)}S_{\beta\alpha}.\label{coefficient}
\end{equation}
 The product in (\ref{coefficient}) is over all inversions $(\alpha, \beta)$ of $\sigma$ and Id represents the identity permutation. We denote the product as
  \begin{equation}
     {A}_{\sigma} = \prod_{(\alpha,\beta)}S_{\beta\alpha}. \label{coefficient2}
  \end{equation}
  In order to find the transition probability, the next step is to have $u(x_N,\cdots,x_1;t)$ obtained through $u^0(x_N,\cdots,x_1;t) $ by Definition \ref{DEF} satisfy the initial condition. As in other models whose transition probabilities are known, the transition probability can be obtained by taking
\begin{equation*}
 A'_{\textrm{Id}}(z_1,\cdots,z_N) = \prod_{i=1}^Nz_i^{-y_i-1}
\end{equation*}
where $(y_N,\cdots,y_1) =Y$ is the initial state and integrating $u(x_N,\cdots,x_1;t)$ over proper contours with respect to $z_1,\cdots,z_N$ such that the initial condition is satisfied.  For some models with simple $S$-matrices \cite{Borodin2,Chatterjee,Rakos2,Schutz}, showing that the initial condition is satisfied is not hard. However, showing that the initial condition is satisfied is a nontrivial task in the ASEP \cite{TW1,TW6} and in the two-sided PushASEP \cite{Lee3} because of the complexity of poles of $S$-matrices. In the $q$-TAZRP, in addition to the complexity of poles of the $S$-matrices,  $u(x_N,\cdots,x_1;t)$ has a factor dependent on $(x_N,\cdots,x_1)$, which makes it more complicated to show that the initial condition is satisfied.
The following lemma is needed for the transition probability and the probability distribution for the left-most particle's position at time $t$. Let $\mathcal{C}_{r}$ be circles centered at the origin with radius $r$ on the complex plane.
\begin{lemma}\label{lemmaTran}
Let $0<r<\frac{1-|q|}{|1-q|}$ and  $0<|q|<1$. Let
\begin{equation*}
F(z_i,z_j) = \frac{E(z_i,z_j)}{z_i -qz_j - (1-q)z_iz_j}
\end{equation*}
be a function on $\mathcal{C}_{r} \times \mathcal{C}_{r}$, where $E(z_i,z_j)$ is an entire function of $z_i,z_j$. Then, (\rmnum{1}) as a function of $z_j$, $F(z_i,z_j)$ has no pole inside $\mathcal{C}_{r}$, and (\rmnum{2})   as a function of $z_i$, $F(z_i,z_j)$ has a pole inside $\mathcal{C}_{r}.$
\end{lemma}
\begin{proof}
As a function of $z_j$, $F(z_i,z_j)$ has  a pole at
\begin{equation*}
z_j = \frac{z_i}{q+(1-q)z_i}.
\end{equation*}
Since $0<r<\frac{1-|q|}{|1-q|}$ and  $0<|q|<1$,
\begin{equation*}
\Bigg|\frac{z_i}{q+(1-q)z_i}\Bigg| > \frac{r}{|q|+|1-q|r} > r.
\end{equation*}
As a function of $z_i$, $F(z_i,z_j)$ has  a pole at
\begin{equation*}
z_i = \frac{qz_j}{1-(1-q)z_j}.
\end{equation*}
Since $0<r<\frac{1-|q|}{|1-q|}$ and  $0<|q|<1$,
\begin{equation*}
\Bigg|\frac{qz_j}{1-(1-q)z_j}\Bigg| < \frac{|q|r}{1-|1-q|r} <r
\end{equation*}
where $1-|1-q|r>0$.
\end{proof}

\begin{lemma}\label{identity1} Let $0<r<\frac{1-|q|}{|1-q|}$ and  $0<|q|<1$.  We have the identity
\begin{equation}
\sum_{\sigma \in \mathbb{S}_n}\Big(\frac{1}{2\pi i}\Big)^n \int_{\mathcal{C}_{r}}\cdots\int_{\mathcal{C}_{r}} {A}_{\sigma} \prod_{i=1}^n z_{i}^{ -1}dz_1\cdots dz_n =[n]_q!.\label{lemmaIdentity}
\end{equation}
\end{lemma}
\begin{proof}
Evaluating the integral in the left-hand side of (\ref{lemmaIdentity}) by using Lemma \ref{lemmaTran}, the integral is equal to $q^{\textrm{inv}(\sigma)}$ for each $\sigma$ where $\textrm{inv}(\sigma)$ is the number of inversions of $\sigma$. Hence, the left-hand side of (\ref{lemmaIdentity}) is equal to $\sum_{\sigma \in \mathbb{S}_n}q^{\textrm{inv}(\sigma)}$ and so (\ref{lemmaIdentity}) simply implies a well-known identity for $[n]_q!$.
\end{proof}
\noindent Alternatively, Lemma \ref{identity1} can be proved by the following identity   and Lemma \ref{lemmaTran}.\footnote{We thank Ivan Corwin for informing us that (\ref{Simple}) is equivalent to (1.4) in \Rmnum{3} in \cite{Macdonald}.}
\begin{lemma}\label{SimpleLemma}
\begin{equation}
\sum_{\sigma \in \mathbb{S}_n} \textrm{sgn}(\sigma)\prod_{i<j} \big(z_{\sigma(j)} - qz_{\sigma(i)}- (1-q)z_{\sigma(i)}z_{\sigma(j)} \big) = [n]_q! \prod_{i<j} (z_j -z_i) \label{Simple}
\end{equation}
\end{lemma}
\begin{theorem}\label{theorem1} Let $0<r<\frac{1-|q|}{|1-q|}$ and  $0<|q|<1$.
The transition probability of the $q$-TAZRP is given by
\begin{equation}
P_Y(X;t)= {W}_X\Big(\frac{1}{2\pi i}\Big)^N \int_{\mathcal{C}_{r}}\cdots\int_{\mathcal{C}_{r}}\sum_{\sigma \in \mathbb{S}_N} {A}_{\sigma} \prod_{i=1}^N z_{\sigma(i)}^{x_i-y_{\sigma(i)} -1}e^{ \epsilon(z_i) t}dz_1\cdots dz_N\label{transition}
\end{equation}
where the \textit{weight} $W_X$ is given by
\begin{equation}
\frac{1}{\prod\limits_{\mathbf{x}\in \Lambda(X)}[\eta(\mathbf{x})]_q!} \label{weight}
\end{equation}
and $A_{\sigma}$ is given by (\ref{coefficient2}) and $\epsilon(z_i)$ is given by (\ref{energe}).
\end{theorem}
\noindent  For the proof of Theorem \ref{theorem1}, it suffices to show that (\ref{transition}) satisfies the initial condition, that is, $P_Y(X;0) = \delta_{XY}$. That (\ref{transition}) satisfies (\ref{Neq}) and (\ref{boundaryN}) can be shown as in the ASEP \cite{TW1}.
\begin{proof} (\textit{Initial condition})  We first show that $P_Y(X;0) =0$ when $X \neq Y$ and $X,Y \in \mathbf{S}$ by induction. When $N=1$, it is obvious that $P_Y(X;0) =0$. Assume that the statement is true for $N-1$. We will show that
\begin{equation}
\int_{\mathcal{C}_{r}}\cdots\int_{\mathcal{C}_{r}}\sum_{\substack{\sigma  \in \mathbb{S}_{N},\\ \sigma(1) =k}} {A}_{\sigma} \prod_{i=1} z_{\sigma(i)}^{x_i-y_{\sigma(i)} -1}dz_1\cdots dz_N =0 \label{initialEQ1}
\end{equation}
when $X\neq Y$  for each $k$. Let us rewrite (\ref{initialEQ1}) as $A\cdot B$ where
\begin{equation}
A=\int_{\mathcal{C}_{r}}\cdots\int_{\mathcal{C}_{r}}\sum_{\substack{\sigma  \in \mathbb{S}_{N},\\ \sigma(1) =k}} \prod_{\substack{(\alpha,\beta),\\ \beta\neq k}} S_{\beta\alpha}\prod_{i=2}^{N} z_{\sigma(i)}^{x_i-y_{\sigma(i)} -1}dz_{\sigma(2)}\cdots dz_{\sigma(N)} \label{A}
\end{equation}
and
\begin{equation}
B=\int_{\mathcal{C}_{r}}\prod_{(\beta,k)}S_{k\beta }  z_{k}^{x_1-y_{k} -1}dz_k. \label{B}
\end{equation}
  If $x_1 > y_k$, then the pole of the integrand of (\ref{B}) which comes from $S_{k\beta }$ is outside of the contour by Lemma \ref{lemmaTran}, which give $B=0$. If $x_1 = y_k$, it implies that $x_1 =\cdots = x_k = y_1 \cdots = y_k$ because $x_i \geq x_{i+1}$ and $y_i \geq y_{i+1}$, and so we consider $k>1$ because we are considering $X \neq Y$. That $X\neq Y$ implies that there exists $j>k$ such that $x_j > y_j$. Therefore,
  \begin{equation*}
   (x_N,\cdots,x_2) \neq (y_N,\cdots,y_{k+1},y_{k-1},\cdots,y_1)
  \end{equation*}
   in $A$ and  the induction hypothesis for $N-1$ implies that $A=0$.\\
\indent Now, we show that $P_Y(Y;0) =1$. Assume that $Y=\big((\mathbf{y}_m)^{n_m},\cdots,(\mathbf{y}_1)^{n_1}\big)$ without loss of generality, and recall that $\sum_i^m n_i =N$.
We will show that
\begin{equation}
\sum_{{\sigma} \in \mathbb{S}_N }\Big(\frac{1}{2\pi i}\Big)^N \int_{\mathcal{C}_{r}}\cdots\int_{\mathcal{C}_{r}} {A}_{\sigma} \prod_{i=1} z_{\sigma(i)}^{y_i-y_{\sigma(i)} -1}dz_1\cdots dz_N = {\prod_{i=1}^{m}[n_i]_q!}. \label{XY}
\end{equation}
 Define a bijection $\tilde{\sigma}(i)$ on $\{1,\cdots,N\}$ by
\begin{equation*}
\tilde{\sigma}(i) = \begin{cases}
                       \sigma_1(i) & \textrm{if}~~i\in \mathbf{N}_1=\{1,\cdots,n_1\}\\
                       \sigma_2(i) & \textrm{if}~~i\in \mathbf{N}_2= \{n_1+1,\cdots,n_1+n_2\}\\
                       \vdots & \vdots\\
                       \sigma_m(i) & \textrm{if}~~i\in \mathbf{N}_m=\big\{ n_1+\cdots+n_{m-1}+1,\cdots, n_1+\cdots+n_m\big\}
                    \end{cases}
\end{equation*}
where each $\sigma_k$ is a bijection on $\mathbf{N}_k$ to $\mathbf{N}_k$ and denote  the set of all $\tilde{\sigma}$ by $\tilde{\mathbb{S}}_N$, which is a subset of $\mathbb{S}_N$. We notice that
\begin{eqnarray}
 \sum_{\tilde{\sigma} \in \tilde{\mathbb{S}}_N } A_{\tilde{\sigma}} \prod_{i=1} z_{\tilde{\sigma}(i)}^{y_i-y_{\tilde{\sigma}(i)} -1} &=& \sum_{\sigma_1,\cdots,\sigma_m}A_{{\sigma}_1}\cdots A_{\sigma_m}\prod_{i=1}^{N} z_{\tilde{\sigma}(i)}^{y_i-y_{\tilde{\sigma}(i)} -1} \nonumber\\
  &=& \Bigg(\sum_{\sigma_1}A_{\sigma_1}\prod_{i=1}^{n_1} z_{\sigma_1(i)}^{y_i-y_{\sigma_1(i)} -1}\Bigg)\Bigg(\sum_{\sigma_2}A_{\sigma_2}\prod_{i=n_1+1}^{n_1+n_2} z_{\sigma_2(i)}^{y_i-y_{\sigma_2(i)} -1}\Bigg) \nonumber \\
 & & \hspace{1cm} \cdots\Bigg(\sum_{\sigma_m}A_{\sigma_m}\prod_{i=\sum_{i=1}^{m-1} n_i+1}^{N} z_{\sigma_m(i)}^{y_i-y_{\sigma_m(i)} -1}\Bigg) \nonumber\\
 &=&\Bigg(\sum_{\sigma_1}A_{\sigma_1}\prod_{i=1}^{n_1} z_{i}^{-1}\Bigg)\Bigg(\sum_{\sigma_2}A_{\sigma_2}\prod_{i=n_1+1}^{n_1+n_2} z_{i}^{-1}\Bigg) \nonumber \\
 & & \hspace{1cm} \cdots\Bigg(\sum_{\sigma_m}A_{\sigma_m}\prod_{i=\sum_{i=1}^{m-1} n_i+1}^{N} z_{i}^{ -1}\Bigg) \label{EQUA}
\end{eqnarray}
 where $A_{\sigma_i}$ is defined by inversions of $\sigma_i$ in the same way as (\ref{coefficient2}) and the last equality is due to the initial condition $Y=\big((\mathbf{y}_m)^{n_m},\cdots,(\mathbf{y}_1)^{n_1}\big)$. Now, integrating (\ref{EQUA}) over contours $\mathcal{C}_{r}$  and multiplying by $\frac{1}{2\pi i}$ for each variable, we obtain the right-hand side of (\ref{XY}) by Lemma \ref{identity1}. Now, we show that for each $\sigma \in \mathbb{S}_N \setminus \tilde{\mathbb{S}}_N$ the integral in the left-hand side of (\ref{XY}) is zero. If $\sigma \notin \tilde{\mathbb{S}}_N$, then there  exists  $j$ such that $j \in \mathbf{N}_l$ but $\sigma(j) \in \mathbf{N}_k$ with $k<l$, and let
 \begin{equation*}
 \mathbf{m} = \sigma(\mathbf{j})= \max \{\sigma (j) : j \in \mathbf{N}_l ~~\textrm{but}~~ \sigma(j) \in \mathbf{N}_k ~~\textrm{with}~~ k<l \}.
 \end{equation*}
Then for such a $\sigma \notin \tilde{\mathbb{S}}_N$, we may write
 \begin{equation}
 \int_{\mathcal{C}_{r}}\cdots\int_{\mathcal{C}_{r}} \frac{A_{\sigma}}{\prod_{(\beta,\mathbf{m})}S_{\mathbf{m}\beta}} \prod_{\sigma(i)\neq \mathbf{m}} z_{\sigma(i)}^{y_i-y_{\sigma(i)} -1}\Big(\prod_{(\beta,\mathbf{m})}S_{\mathbf{m}\beta}~z_{\mathbf{m}}^{y_{\mathbf{j}} - y_{\mathbf{m}}-1}\Big)dz_1\cdots dz_N. \label{proofeq}
\end{equation}
 Then $z_{\mathbf{m}}$ variable appears only in the parenthesis in (\ref{proofeq}), and so integrating with respect to $z_{\mathbf{m}}$, the integral is zero by Lemma \ref{lemmaTran} and the fact that $y_{\mathbf{j}} > y_{\mathbf{m}}.$
\end{proof}
\begin{remark} It was pointed out by Ivan Corwin that showing that the initial condition for Theorem \ref{theorem1} is satisfied  is equivalent to showing that $\mathcal{K}^{\textrm{q-Boson}} = \textrm{Id}$ which implies a Plancherel formula related to the $q$-TAZRP \cite{Corwin2}.
See \cite[Theorem 3.5.]{Borodin4} to compare the approach we took to prove with the approach in \cite{Borodin4}.
\end{remark}
\subsection{On the mapping between the $q$-TAZRP and the $q$-TASEP}
The $q$-TASEP was introduced by Borodin and Corwin in the formalism of Macdonald processes \cite{Borodin3}.  According to the definition by Borodin and Corwin, the $q$-TASEP is a continuous-time interacting particle system on $\mathbb{Z}$ where the $i$th right-most particle jumps to the right by one independently of the others according to an exponential clock with rate $a_i(1-q^{x_{i-1} - x_i -1})$ where $a_i>0$ and $0< q <1$ and $x_i$ is the position of the $i$th right-most particle. That is, the jumping rate of the $i$th right-most particle depends on the gap between the $i$th particle and $(i-1)$th particle when $a_i=1$. Also, there is a mapping between the $N$-particle $q$-TASEP and the $(N+1)$-site $q$-TAZRP and the duality of these two models was studied by Borodin, Corwin and Sasamoto \cite{Borodin5}. However, the mapping between the $N$-particle $q$-TASEP and the $(N+1)$-site $q$-TAZRP  is not a bijection. Nevertheless, there is a bijection from the state space $\mathbf{S}$ of the $N$-particle \textit{ZRP-like} model to the state space $\mathbf{S}'$ of the $N$-particle \textit{ASEP-like} model \cite{Lee2,Povol3} :
\begin{equation}
  T: (x_i) \in \mathbf{S} \mapsto (x_i -i) \in \mathbf{S}' \label{bijection}
\end{equation}
where $x_i$ is the $i$th right-most particle's position.
 It is reasonable that  the transition probabilities of two models should be conserved under this mapping, that is,
 \begin{equation*}
    P_Y(X;t) = P_{T(Y)}(T(X);t).
 \end{equation*}
 In fact, this conservation can be confirmed by direct computation of two transition probabilities but the mapping (\ref{bijection}) is predicted at the level of finding $S$-matrices of two models (See \cite{Lee2} for the TAZRP with constant rates and the TASEP.). Under the mapping $T$, the $N$-particle $q$-TAZRP is mapped to a variant of Borodin and Corwin's $q$-TASEP,  where the jumping rate of a particle depends on the size of the cluster of particles attached to the back of the particle. For this model, the two-particle $S$-matrix is given by
 \begin{equation*}
  -\frac{z_{\alpha}}{z_{\beta}}\cdot \frac{z_{\beta} - qz_{\alpha} - (1-q)z_{\alpha}z_{\beta}}{z_{\alpha} - qz_{\beta} - (1-q)z_{\alpha}z_{\beta}}, ~~z_{\alpha},z_{\beta} \in \mathbb{C}
\end{equation*}
which is to be compared with (\ref{SMatrix}). Recalling Lemma 3 in \cite{Lee2}, then the Bethe ansatz solution corresponding to (\ref{solution}) is written as
\begin{equation}
 \sum_{\sigma \in \mathbb{S}_N} {A}_{\sigma} \prod_{i=1}^N z_{\sigma(i)}^{x_i-i + \sigma(i) - y_{\sigma(i)}-1}e^{\epsilon(z_i) t} \label{variant}
\end{equation}
where ${A}_{\sigma}$ is the same as that of the $q$-TAZRP after taking $A'_{\textrm{Id}} = \prod_{i=1}^Nz_i^{-y_i-1}.$  Obviously, we can read off the mapping $T$ in (\ref{variant}). However, if we are concerned with an infinite system, the $q$-TASEP and its variant are equivalent by the \textit{particle-hole} symmetry.
\section{The probability for the left-most particle's position }
The probability that the left-most particle is at $x$ at time $t$, $\mathbb{P}_Y(x_N(t) = x)$ is computed by summing transition probabilities over all possible states
$$(x,x+k_1,\cdots,x+k_1+\cdots + k_{N-1}),~~ k_i \in \mathbb{Z}_+.$$ That is,
\begin{eqnarray}
\mathbb{P}_Y(x_N(t) = x) & = &\sum_{k_1,\cdots,k_{N-1}=0}^{\infty} P_Y(x,x+k_1,\cdots,x+k_1+\cdots + k_{N-1};t)  \label{LEFTMOST} \\
                 & =& \Big(\frac{1}{2\pi i}\Big)^N \int_{\mathcal{C}_{r}}\cdots\int_{\mathcal{C}_{r}}\sum_{\sigma \in \mathbb{S}_N} A_{\sigma}F_{\sigma}^{(N)} \nonumber \\
                 & & \hspace{3cm} \times\Big(1-\prod_i^Nz_i\Big) \prod_{i=1}^N z_{i}^{x-y_{i} -1}e^{ \epsilon(z_i) t}dz_1\cdots dz_N\nonumber
\end{eqnarray}
where
\begin{eqnarray}
F_{\sigma}^{(N)} &:=& F(z_{\sigma(1)},\cdots,z_{\sigma(N)}) \nonumber  \\
 &=& \frac{1}{1-\prod_{i=1}^N z_i}\sum_{k_1,\cdots,k_{N-1}=0}^{\infty} (\textrm{\textit{weight}}) \cdot \prod_{j=1}^{N-1}\Bigg(\prod_{i=1}^{N-j} z_{\sigma(i)}\Bigg)^{k_{j}}.\label{series}
\end{eqnarray}
We find the explicit form of $F_{\sigma}^{(N)}$.
Since the \textit{weight} is a function of $k_1,\cdots,k_{N-1}$ in this case, we denote it by $W(k_1,\cdots,k_{N-1})$.  However, note that $W(k_1,\cdots,k_{N-1})$ depends only on whether each $k_i$ is zero or not (Recall (\ref{weight}).)  That is, for  $(k_1,\cdots,k_{N-1})$ and $(k_1',\cdots,k'_{N-1})$ if either $k_i = k'_i=0$ or $k_i,k'_i \neq 0$ for each $i$ , then $W(k_1,\cdots,k_{N-1})= W(k'_1,\cdots,k'_{N-1})$. The geometric series over indices $k_{j}$ in (\ref{series}) converge because the radius of each contour is less than 1. Let us denote the geometric series after summing over $k_i$ from  $k_{i} =1$ to $\infty$  by
\begin{equation}
\label{Tfunction} t_{i} = \frac{z_{\sigma(1)}\cdots z_{\sigma(N-i)}}{1- z_{\sigma(1)}\cdots z_{\sigma(N-i)}}.
\end{equation}
Defining
\begin{equation*}
\mathbb{I}(k_i) := \begin{cases} 1&\textrm{~~~if~~~} k_i =0 \\
t_{i}&\textrm{~~~if~~~} k_i = 1,
\end{cases}
\end{equation*}
the right-hand side of  (\ref{series}) is the sum of $2^{N-1}$ terms
\begin{equation}
\frac{1}{1-\prod_{i=1}^Nz_i} \sum_{k_1,\cdots,k_{N-1}=0,1}W(k_1,\cdots,k_{N-1})\mathbb{I}(k_1)\cdots \mathbb{I}(k_{N-1}). \label{sum3}
\end{equation}
It is more convenient to represent the sum (\ref{sum3}) over $k_1,\cdots,k_{N-1} =0,1$ by the sum over all \textit{compositions} of $N$. Notice that each $(k_1,\cdots,k_{N-1})$ where $k_i=0$ or $1$ is bijectively mapped to a \textit{composition} of $N$. For example, $(k_1,k_2,k_3) = (1,0,1)$ which represents the state $(x,x+1,x+1,x+2) =((x)^1,(x+1)^2,(x+2)^1)$ is mapped to a \textit{composition} $(1,2,1)$ of 4, and in this case the \textit{weight} (\ref{weight}) is given by $\frac{1}{[1]_q![2]_q![1]_q!}$. Hence,  there exist a unique $(k_1,\cdots,k_{N-1})$ where $k_i=0 ~\textrm{or}~1$ for a composition $(m_1,\cdots,m_n)$ of $N$. To be concrete, $(m_1,\cdots,m_n)$  corresponds to $(k_1,\cdots,k_{N-1})$ such that all $k_i = 0$ except $i=m_1,m_1+m_2,\cdots,m_1+m_2+ \cdots + m_{n-1}$. For this $(k_1,\cdots,k_{N-1})$,
\begin{equation*}
\mathbb{I}(k_1)\cdots \mathbb{I}(k_{N-1}) = t_{m_1}t_{m_1+m_2}\cdots t_{m_1+m_2+\cdots+m_{n-1}}
\end{equation*}
holds  when $n \geq2$.  If $n=1$, that is, $(m_1) =(N)$ which corresponds to  $(k_1,\cdots,k_{N-1}) = (0,\cdots,0)$, then $\mathbb{I}(k_1)\cdots \mathbb{I}(k_{N-1})=1$. Recalling that the \textit{weight} for a composition $(m_1,\cdots,m_n)$ of $N$ is $\frac{1}{[m_1]_q!\cdots [m_n]_q!}$, we have \begin{equation}
F_{\sigma}^{(N)} = \frac{1}{1-\prod_{i=1}^Nz_{\sigma(i)}}\sum_{(m_1,\cdots,m_n)}\frac{1}{[m_1]_q!\cdots [m_n]_q!}~t_{m_1}t_{m_1+m_2}\cdots t_{m_1+\cdots+m_{n-1}} \label{series2}
\end{equation}
where the sum is over all \textit{compositions} of $N$, understanding that $t_i$ is a $\sigma$-dependent variable.
\begin{lemma} \label{identityLemma}
  Let $T_i$ be a function on $\mathbb{S}_N$ to $\mathbb{S}_N$ defined by $T_i\sigma = \tilde{\sigma}$ such that $\tilde{\sigma}(N-i) = \sigma(N-i+1), \tilde{\sigma}(N-i+1) = \sigma(N-i)$ and $\tilde{\sigma}(k) = \sigma(k)$ for $k \neq N-i,N-i+1$, and set $T_0$ to be the identity.  Then for a fixed $\sigma \in \mathbb{S}_N$, the following holds.
\begin{equation}
\sum_{i=0}^{N-1}q^iF_{(T_i\cdots T_0)\sigma}^{(N)}\big|_{z_{\sigma(N)}=0}= F_{\sigma'}^{(N-1)} \label{LEMMA}
\end{equation}
where $\sigma'$ is a bijection on $\{1,\cdots,N-1\}$ to $\{1,\cdots,N\} \setminus \{\sigma(N)\}$.
\end{lemma}
\begin{proof} Let us denote $t_i$ for $T_k\cdots T_0 \sigma$ in (\ref{Tfunction}) by $t_i^{(k)}$. Note that a composition $(m_1,\cdots,m_n)$ of $N-1$ is  mapped to a pair of two compositions of $N$,  $(m_1+1,m_2,\cdots,m_n)$ and $(1,m_1,\cdots,m_n)$.
We will show that the term for $(m_1,\cdots,m_n)$ in the right-hand side of (\ref{LEMMA}) is equal to the sum of terms for $(m_1+1,m_2,\cdots,m_n)$
and terms for $(1,m_1,\cdots,m_n)$ in the left-hand side of (\ref{LEMMA}).  The sum of two terms for $(m_1+1,m_2,\cdots,m_n)$
and $(1,m_1,\cdots,m_n)$  in $F_{T_i\cdots T_0\sigma}^{(N)}$ is
\begin{eqnarray}
& &\frac{1}{1-z_1 \cdots z_N}\Bigg(\frac{1}{[m_1+1]_q![m_2]_q!\cdots [m_n]_q!}~t_{m_1+1}^{(i)}t_{m_1+m_2+1}^{(i)}\cdots t_{m_1+\cdots+m_{n-1}+1}^{(i)} \nonumber  \\
&+ & t_1^{(i)}\frac{1}{[m_1]_q!\cdots [m_n]_q!}~t_{m_1+1}^{(i)}t_{m_1+m_2+1}^{(i)}\cdots t_{m_1+\cdots+m_{n-1}+1}^{(i)}\Bigg).\label{LemmaEq5}
\end{eqnarray}
When $i=0$, substituting $z_{\sigma(N)} =0$ in (\ref{LemmaEq5}), it  becomes
\begin{eqnarray}
& &\frac{1}{[m_1+1]_q![m_2]_q!\cdots [m_n]_q!}~t_{m_1+1}^{(0)}t_{m_1+m_2+1}^{(0)}\cdots t_{m_1+\cdots+m_{n-1}+1}^{(0)} \nonumber  \\
&+ & t_1^{(0)}\frac{1}{[m_1]_q!\cdots [m_n]_q!}~t_{m_1+1}^{(0)}t_{m_1+m_2+1}^{(0)}\cdots t_{m_1+\cdots+m_{n-1}+1}^{(0)}.\label{LemmaEq55}
\end{eqnarray}
Substituting $z_{\sigma(N)} =0$ in (\ref{LemmaEq5}) for $0<i\leq m_1$, the term for $(m_1+1,\cdots,m_n)$   is
\begin{eqnarray}
& &\frac{1}{[m_1+1]_q![m_2]_q!\cdots [m_n]_q!}~t_{m_1+1}^{(i)}t_{m_1+m_2+1}^{(i)}\cdots t_{m_1+\cdots+m_{n-1}+1}^{(i)} \nonumber\\
&=&\frac{1}{[m_1+1]_q![m_2]_q!\cdots [m_n]_q!}~t_{m_1+1}^{(0)}t_{m_1+m_2+1}^{(0)}\cdots t_{m_1+\cdots+m_{n-1}+1}^{(0)} \label{LemmaEq6}
\end{eqnarray}
where the equality is due to $t_j^{(0)} = t_j^{(1)} = \cdots = t_j^{(j-1)}$ and the term for $(1,m_1,\cdots,m_n)$ is zero because $t_1^{(i)}=0$ for all $i \geq1$. Hence, substituting $z_{\sigma(N)} =0$ in $F_{T_i\cdots T_0\sigma}^{(N)}$ for $0<i\leq m_1$, (\ref{LemmaEq5}) becomes
  \begin{equation}
     \frac{1}{[m_1+1]_q![m_2]_q!\cdots [m_n]_q!}~t_{m_1+1}^{(0)}t_{m_1+m_2+1}^{(0)}\cdots t_{m_1+\cdots+m_{n-1}+1}^{(0)}.\label{LemmaEq666}
  \end{equation}
  If $i > m_1$, then $t_{m_1+1}^{(i)} =0$ when $z_{\sigma(N)} =0$, and so (\ref{LemmaEq5}) is zero. Therefore, recalling that
  $$1+q+\cdots + q^{m_1} = [m_1+1]_q,$$
    the left-hand side of (\ref{LEMMA}) is given by
\begin{eqnarray*}
& &\frac{1}{1-z_{\sigma(1)}\cdots z_{\sigma(N-1)}}\frac{1}{[m_1]_q!\cdots [m_n]_q!}~t_{m_1+1}^{(0)}t_{m_1+m_2+1}^{(0)}\cdots t_{m_1+\cdots+m_{n-1}+1}^{(0)},
\end{eqnarray*}
where $t_i^{(0)}$s are for $\sigma$, and so replacing $t_i^{(0)}$s by the notations  for $\sigma'$, we have
\begin{eqnarray*}
& &\frac{1}{1-z_{\sigma'(1)}\cdots z_{\sigma'(N-1)}}\frac{1}{[m_1]_q!\cdots [m_n]_q!}~t_{m_1}t_{m_1+m_2}\cdots t_{m_1+\cdots+m_{n-1}},
\end{eqnarray*}
where $t_i$ are defined for $\sigma'$. This is the term for $(m_1,\cdots,m_n)$ in the right-hand side of (\ref{LEMMA}). This completes the proof.
\end{proof}
\begin{proposition}\label{NewIdentity} We have
\begin{equation}
 \sum_{\sigma \in \mathbb{S}_N} \textrm{sgn}(\sigma) \prod_{1 \leq i<j\leq N}\big(z_{\sigma(i)} - q z_{\sigma(j)} - (1-q)z_{\sigma(i)}z_{\sigma(j)}\big)F_{\sigma}^{(N)}= \frac{\prod\limits_{1 \leq i<j\leq N}(z_i - z_j)}{\prod\limits_{i=1}^N (1-z_i)}.\label{identity3}
\end{equation}
\end{proposition}
\begin{proof}
We will prove
\begin{eqnarray}
 & &{\prod_{i=1}^N (1-z_i)}\sum_{\sigma \in \mathbb{S}_N} \textrm{sgn}(\sigma) \prod_{1 \leq i<j\leq N}\big(z_{\sigma(i)} - q z_{\sigma(j)} - (1-q)z_{\sigma(i)}z_{\sigma(j)}\big)F_{\sigma}^{(N)}\nonumber \\
 & & \hspace{2cm}= {\prod_{1 \leq i<j\leq N}(z_i - z_j)} \hspace{2cm}\textrm{for}~~z_1,\cdots,z_N \neq 1 \label{identity4}
\end{eqnarray}
   by induction. It can be shown that  (\ref{identity4}) is true for $N=2$ by direct computation. Assume that (\ref{identity4}) holds for $N-1$. Since the left-hand side of (\ref{identity4}) is an antisymmetric function of $z_1,\cdots,z_N$, it is divisible by the Vandermonde determinant, that is, the left-hand side of (\ref{identity4}) must be written as
 \begin{equation}
 G(z_1,\cdots,z_N)\prod_{1 \leq i<j \leq N}(z_i -z_j) \label{prop-eq1}
 \end{equation}
  where $G(z_1,\cdots,z_N)$ is a symmetric function of $z_1,\cdots,z_N$. We will show that $G(z_1,\cdots,z_N)=1.$ Let $\sigma_{\alpha}$ be a permutation on $\{1,\cdots,N\}$ such that
  \begin{equation*}
      \sigma_{\alpha}(i) = \begin{cases}
                       {\alpha} & \textrm{~if~} i=N\\
                       \sigma'(i) & \textrm{~if~} i \neq N
                    \end{cases}
  \end{equation*}
  where $\sigma'$ is a permutation on $\{1,\cdots,N-1\}$ to $\{1,\cdots,N\} \setminus \{\alpha\}$. Then the left-hand side of (\ref{identity4}) is written as
  \begin{eqnarray*}
   & & {\prod_{i=1}^N (1-z_i)} \sum_{\sigma' \in \mathbb{S}_{N-1}} \textrm{sgn}(\sigma')(-1)^{N-\alpha} \sum_{k=0}^{N-1}(-1)^k\\
    & & \hspace{1cm} \times  \prod_{1 \leq i<j\leq N}\big(z_{T_k\cdots T_0\sigma_{\alpha}(i)} - q z_{T_k\cdots T_0\sigma_{\alpha}(j)} - (1-q)z_{T_k\cdots T_0\sigma_{\alpha}(i)}z_{T_k\cdots T_0\sigma_{\alpha}(j)}\big)F_{T_k\cdots T_0\sigma_{\alpha}}^{(N)},
  \end{eqnarray*}
  where $(-1)^k$ implies the sign is alternating and $T_k$ is defined as in Lemma \ref{identityLemma}.
  Substituting 0 for $z_{\alpha}$, for each $k$, the second product yields a constant $(-q)^k$, and thus we have
   \begin{equation*}
   \prod_{i\neq {\alpha}}^Nz_i \prod_{i\neq {\alpha}}^N(1-z_i)\sum_{\sigma' \in \mathbb{S}_{N-1}} \textrm{sgn}(\sigma')(-1)^{N-\alpha}\prod_{1 \leq i<j\leq N-1}\big(z_{\sigma'(i)} - q z_{\sigma'(j)} - (1-q)z_{\sigma'(i)}z_{\sigma'(j)}\big)F_{\sigma'}^{(N-1)}
  \end{equation*}
by Lemma \ref{identityLemma}. This  is equal to
\begin{equation*}
 (-1)^{N-\alpha}\prod_{i\neq \alpha}^Nz_i {\prod\limits_{\substack{i<j,\\ i,j\neq \alpha}}(z_i - z_j)}
\end{equation*}
by induction hypothesis. Finally, by substituting 0 for $z_{\alpha}$ in the right-hand side of (\ref{prop-eq1}), we have
\begin{equation*}
G(z_1,\cdots,z_N)\big|_{z_{\alpha}=0} (-1)^{N-\alpha}\prod_{i\neq \alpha}^Nz_i {\prod\limits_{\substack{i<j,\\ i,j\neq \alpha}}(z_i - z_j)}
\end{equation*}
 and so we conclude that
\begin{equation}
G(z_1,\cdots,z_N)\big|_{z_{\alpha}=0} = 1. \label{symmetric}
\end{equation}
 Since $\alpha$ is arbitrary in $\{1,\cdots,N\}$, the only possible symmetric function $G(z_1,\cdots,z_N)$ that satisfies (\ref{symmetric}) is $G(z_1,\cdots,z_N) = 1.$ This completes the proof.
\end{proof}
\begin{remark} The method used in the proof of  Proposition \ref{NewIdentity} can be an alternate way to prove the identity (1.6) in \cite{TW1} where the residue computation was used.
\end{remark}
\noindent By (\ref{LEFTMOST}) and Proposition \ref{NewIdentity} the following theorem is immediately obtained.
\begin{theorem}\label{THEO} Let $0<r<\frac{1-|q|}{|1-q|}$ and  $0<|q|<1$. For the $q$-TAZRP, we have
\begin{equation}
\mathbb{P}_Y(x_N(t) = x) = \Big(\frac{1}{2\pi i}\Big)^N\int_{\mathcal{C}_{r}}\cdots \int_{\mathcal{C}_{r}} I(z_1,\cdots,z_N) \prod_{i=1}^Nz_i^{x-y_i-1}e^{\epsilon (z_i)t}dz_1\cdots dz_N \label{THEORME2}
\end{equation}
where
\begin{equation*}
I(z_1,\cdots,z_N) = \prod\limits_{i<j}^N\frac{z_i - z_j}{z_i - qz_j - (1-q)z_iz_j}\frac{1-\prod\limits_{i=1}^Nz_i}{\prod\limits_{i=1}^N(1-z_i)}
\end{equation*}
and $\epsilon(z_i)$ is given by (\ref{energe}).
\end{theorem}
\noindent The following corollary states that  provided that the initial state $Y=(y,\cdots,y)$  the integrand of $\mathbb{P}_{(y,\cdots,y)}(x_N(t) \geq x)$ is represented by a determinant.
\begin{corollary}\label{COROLL}
For the $q$-TAZRP with $0<|q|<1$,  we have
\begin{equation}
\mathbb{P}_{(y,\cdots,y)}(x_N(t) \geq x) = C_N\int_{\mathcal{C}_{r}}\cdots \int_{\mathcal{C}_{r}} \det \big(K(z_i,z_j)\big) dz_1\cdots dz_N \label{Distribution}
\end{equation}
where
\begin{equation*}
K(z_i,z_j) = \frac{z_i^{x-y} e^{\epsilon (z_i)t}}{z_i -qz_j - (1-q)z_iz_j}
\end{equation*}
and
\begin{equation*}
C_N= \Big(\frac{1}{2\pi i}\Big)^N \frac{[N]_q!}{N!}\frac{(1-q)^N}{q^{N(N-1)/2}}.
\end{equation*}
\end{corollary}
\begin{proof}
First, summing $\mathbb{P}_{(y,\cdots,y)}(x_N(t) = x+i)$ over  $i \in \mathbb{Z}_+$ to obtain $\mathbb{P}_{(y,\cdots,y)}(x_N(t) \geq x)$ and anti-symmetrizing the integrand of $\mathbb{P}_{(y,\cdots,y)}(x_N(t) \geq x)$,
 \begin{eqnarray*}
\mathbb{P}_Y(x_N(t) \geq  x) &=&\frac{[N]_q!}{N!} \Big(\frac{1}{2\pi i}\Big)^N\int_{\mathcal{C}_{r}}\cdots \int_{\mathcal{C}_{r}}\prod_{i=1}^Nz_i^{x-y-1}e^{\epsilon (z_i)t} \\
 & &\times\prod_{i \neq j} \frac{z_i -z_j}{z_i - qz_j - (1-q)z_iz_j}\frac{1}{\prod_{i=1}^N(1-z_i)}dz_1\cdots dz_N
\end{eqnarray*}
by using Lemma \ref{SimpleLemma}. By using
\begin{eqnarray*}
& & \det \Bigg(\frac{1}{z_i -qz_j -(1-q)z_iz_j}\Bigg)_{1 \leq i,j \leq N} \\
&=& \frac{q^{N(N-1)/2}}{(1-q)^N}\frac{1}{\prod_{i=1}^Nz_i(1-z_i)}\prod\limits_{i\neq j}\frac{z_i-z_j}{z_i - qz_j - (1-q)z_iz_j}, ~~q \neq 0,1
\end{eqnarray*}
which can be shown by the Cauchy determinant formula and the change of variable $z_i \rightarrow 1/(1+z_i)$, we have (\ref{Distribution}).
 \end{proof}
 \begin{remark} If $q \in (0,1)$ and the contour for $z_i$ variable is $\mathcal{C}_{r_i}$ with $0<r_N<\cdots<r_1 <1$, Theorem \ref{THEO} and  Corollary \ref{COROLL} still hold  and agree with formulas for the $N$-particle $q$-TASEP in \cite{Borodin5} under a suitable change of variables. (See Corollary 2.12, in particular, in \cite{Borodin5}.)
 \end{remark}
\noindent\textbf{Acknowledgement} \\
The authors would like to thank Alexei Borodin and Ivan Corwin for valuable comments on the earlier version of the manuscript. We also would like to express our gratitude to Antti Kupiainen for helpful comments and unstinting support. E. Lee was partially supported  by European Research Council Advanced Grant, the Natural Sciences and Engineering Research Council of Canada (NSERC), the Fonds de recherche du Qu\'{e}bec - nature et technologies  (FQRNT) and the CRM Laboratoire de Physique Math\'{e}matique. M. Korhonen was supported by Academy of Finland.

\end{document}